\documentclass[11pt]{article}
\usepackage[dvips]{graphicx}
\usepackage{amssymb}
\usepackage{amsmath, amsthm}
\usepackage{setspace}
\usepackage{subfig}

\newenvironment{packed_enum}{
\begin{enumerate}
  \setlength{\itemsep}{1pt}
  \setlength{\parskip}{0pt}
  \setlength{\parsep}{0pt}
}{\end{enumerate}}

\newtheorem{definition}{Definition}
\newtheorem{proposition}{Proposition}

\newtheorem{assumption}{Assumption}
\newtheorem{lemma}{Lemma}
\newtheorem{theorem}{Theorem}

\begin{document}

\title{\vspace{-1.5cm}Designing Incentive Schemes Based on Intervention:\\ The Case of Perfect Monitoring}

\author{Jaeok Park and Mihaela van der Schaar\thanks{Electrical Engineering
Department, University of California, Los Angeles (UCLA). Email: \{jaeok, mihaela\}@ee.ucla.edu.}}

\date{}

\maketitle

\vspace{-1cm}

\begin{abstract}
This paper studies a class of incentive schemes based on intervention,
where there exists an intervention device that is able to monitor the actions of users
and to take an action that affects the payoffs of users. We consider the case
of perfect monitoring, where the intervention device can immediately observe the actions of users
without errors. We also assume that there exist
actions of the intervention device that are most and least preferred
by all the users and the intervention device, regardless of the actions of users.
We derive analytical results
about the outcomes achievable with intervention, and illustrate
our results with an example based on the Cournot model.
\end{abstract}

\section{Introduction}

This paper studies incentive schemes to drive self-interested users
toward the system objective. The operation of networks by non-cooperative, self-interested
users in general leads to a suboptimal performance \cite{Dubey}. As a
result, different forms of incentive schemes to improve the performance
have been investigated in the literature.
One form of incentive schemes widely studied in economics and engineering is pricing (or more
generally, transfer of utilities) \cite{varian}. Pricing can induce efficient
use of network resources by aligning private incentives with social objectives.
Although pricing has a solid theoretical foundation, implementing a pricing
scheme can be impractical or cumbersome in some cases. Let us consider a wireless
Internet service as an example. A service provider can limit access
to its network resources by charging an access fee. However, charging an access fee
requires a secure and reliable method to process payments, which creates
burden on both sides of users and service providers. There also arises
the issue of allocative fairness when a service provider charges for the Internet service.
In the presence of the income effect, uniform pricing will bias the allocation
of network resources towards users with high incomes. Because the Internet
can play the role of an information equalizer, it has
been argued in a public policy debate that access to the Internet
should be provided as a public good by a public authority rather than
as a private good in a market \cite{Hallgren}.

Another method to provide incentives is to use repeated interaction \cite{mailath}.
Repeated interaction can encourage cooperative behavior by adjusting
future payoffs depending on current behavior. A repeated game strategy can
form a basis of an incentive scheme in which monitoring and punishment burden
is decentralized to users (see, for example, \cite{La}). However, implementing a repeated game strategy
requires repeated interaction among users, which may not be available.
For example, users interacting in a mobile network change frequently in nature.

In this paper, we study an alternative form of incentive schemes based on intervention, which
was proposed in our previous work \cite{jpark}. In an incentive scheme based on
intervention, a network is augmented with an intervention device that is able to monitor the actions of users
and to take an action that affects the payoffs of users.
Intervention directly affects the
network usage of users, unlike pricing which uses an outside instrument
to affect the payoffs of users. Thus, an incentive scheme based on intervention can provide
an effective and robust method to provide incentives in that users
cannot avoid intervention as long as they use network resources.
Moreover, it does not require long-term relationship among users, which
makes it applicable to networks with a dynamically changing user population.

As a first step toward the study of incentive schemes based on intervention, we focus
in this paper on the case of perfect monitoring, where the intervention device can immediately observe
the actions chosen by users without errors. We derive analytical results
assuming that there exist actions of the intervention device that are most and least preferred
by all the users and the intervention device, regardless of the actions of users. We then illustrate our results
with an example based on the Cournot model.

\section{Model}

We consider a network where $N$ users and an intervention device
interact. The set of the users is denoted by $\mathcal{N} = \{1,\ldots,N\}$.
The action space of user $i$ is denoted by $A_i$,
and the action of user $i$ is denoted by $a_i \in A_i$,
for all $i \in \mathcal{N}$.
An action profile is represented by
a vector $\mathbf{a} = (a_1, \ldots, a_N) \in A \triangleq \prod_{i \in \mathcal{N}} A_i$.
An action profile of the users other than user~$i$ is written as $\mathbf{a}_{-i} = (a_1,\ldots,
a_{i-1},a_{i+1}, \ldots,a_N)$ so that $\mathbf{a}$ can be expressed as $\mathbf{a} = (a_i,\mathbf{a}_{-i})$.
The intervention device observes the actions chosen by
the users immediately, and then it chooses its own action.
The action space of the intervention device is denoted by $A_0$,
and its action is denoted by $a_0 \in A_0$.
For convenience, we sometimes call the intervention device user 0.
The set of the users and the intervention device is denoted by
$\mathcal{N}_0 = \mathcal{N} \cup \{0\}$.

The actions of the intervention device and the users jointly determine
their payoffs. The payoff function of user $i \in \mathcal{N}_0$ is
denoted by $u_i: A_0 \times A \rightarrow \mathbb{R}$. That is,
$u_i(a_0,\mathbf{a})$ represents the payoff that user $i$ receives when the
intervention device chooses action $a_0$ and the users choose an
action profile $\mathbf{a}$. In particular, the payoff of the intervention device,
$u_0(a_0,\mathbf{a})$, can be interpreted as the system objective.
Since the intervention device can choose
its action knowing the actions chosen by the users, a strategy
for it can be represented by a function $f:A \rightarrow A_0$,
which is called an intervention rule. The set of all possible intervention
rules is denoted by $F$.

Suppose that there is a network
manager who determines the intervention rule used by the intervention
device. We assume that the manager can commit to an intervention rule,
for example, by using a protocol embedded in the intervention device.
The game played by the manager and the users is called an intervention game.
The sequence of events in an intervention game can be listed as follows.
\begin{packed_enum}
\item The manager chooses an intervention rule $f \in F$.
\item The users choose their actions $\mathbf{a} \in A$, knowing
the intervention rule $f$ chosen by the manager.
\item The intervention device observes the action profile $\mathbf{a} \in A$
and takes an action $a_0 = f(\mathbf{a}) \in A_0$.
\end{packed_enum}

The payoff function of user $i \in \mathcal{N}_0$ provided that the manager
has chosen an intervention rule $f$ is given by $v_i^f:A \rightarrow \mathbb{R}$,
where
\begin{align}
v_i^f(\mathbf{a}) = u_i(f(\mathbf{a}),\mathbf{a}).
\end{align}
An intervention rule $f$ induces
a simultaneous game played by the users, whose normal form
representation is given by
\begin{align}
\Gamma_f = \left\langle N, (A_i)_{i \in \mathcal{N}}, (v_i^f)_{i \in \mathcal{N}}
\right\rangle.
\end{align}
We can predict actions chosen by the users given an
intervention rule $f$ by applying the solution concept of Nash
equilibrium to the induced game $\Gamma_f$.

\begin{definition}
An intervention rule $f \in {F}$ \emph{sustains} an action profile
$\mathbf{a}^* \in A$ if $\mathbf{a}^*$ is a Nash equilibrium of
the game $\Gamma_f$, i.e.,
\begin{align}
v_i^f(\mathbf{a}^*) \geq v_i^f(a_i,\mathbf{a}_{-i}^*) \quad \text{for all $a_i \in
A_i$, for all $i \in \mathcal{N}$}.
\end{align}
An action profile $\mathbf{a}^*$ is \emph{sustainable} if there exists
an intervention rule $f$ that sustains $\mathbf{a}^*$.
\end{definition}

Let $\mathcal{E}(f) \subseteq A$ be the set of action profiles sustained by $f$.
Then the set of all sustainable action profiles is given by $\mathcal{E} = \cup_{f \in {F}}
\mathcal{E}(f)$. A pair of an intervention rule $f$ and an action profile $\mathbf{a}$ is said
to be attainable if $f$ sustains $\mathbf{a}$. The manager's problem is to
find an attainable pair that maximizes the payoff of the intervention
device among all attainable pairs.

\begin{definition} \label{def:ie}
$(f^*,\mathbf{a}^*) \in {F} \times A$ is an \emph{intervention equilibrium}
if $\mathbf{a}^* \in \mathcal{E}(f^*)$ and
\begin{align}
v_0^{f^*}(\mathbf{a}^*) \geq v_0^f(\mathbf{a})
\end{align}
for all $(f,\mathbf{a}) \in {F} \times A$ such that $\mathbf{a} \in \mathcal{E}(f)$.
$f^* \in {F}$ is an \emph{optimal intervention rule} if there exists
an action profile $\mathbf{a}^* \in A$ such that $(f^*,\mathbf{a}^*)$ is an
intervention equilibrium.
\end{definition}

Intervention equilibrium is a solution concept for intervention
games, based on a backward induction argument. An intervention equilibrium
can be considered as a subgame
perfect equilibrium applied to an intervention game, since
the induced game $\Gamma_f$ is a subgame of an intervention game. It is
implicitly assumed that the manager can induce the users
to choose the best Nash equilibrium for the system in case of multiple Nash
equilibria. One possible explanation for this is that the manager recommends
to the users an action profile sustained by the intervention rule he chooses
so that the action profile becomes a focal point \cite{fudenberg}.
The manager's problem of finding an optimal intervention
rule can be expressed as
\begin{align}
\max_{f \in F} \max_{\mathbf{a} \in \mathcal{E}(f)} v_0^f(\mathbf{a}).
\end{align}

\section{Analytical Results}

In this section, we derive analytical results about sustainable action
profiles and intervention equilibria imposing the following assumption.

\begin{assumption} \label{ass:permon}
There exist $\underline{a}_0, \overline{a}_0 \in A_0$ such that
for all $i \in \mathcal{N}_0$,
\begin{align} \label{eq:a0order}
u_i(\underline{a}_0,\mathbf{a}) \geq u_i(a_0,\mathbf{a}) \geq u_i(\overline{a}_0,\mathbf{a})
\quad \text{for all $a_0 \in A_0$, for all $\mathbf{a} \in A$}.
\end{align}
\end{assumption}

$\underline{a}_0$ and $\overline{a}_0$ can be interpreted as the
minimal and maximal intervention actions of the intervention device,
respectively. For given $\mathbf{a} \in A$, the users and the intervention device
receive the highest (resp. lowest) payoff when the intervention device takes the minimal
(resp. maximal) intervention action. This allows the intervention device to
reward or punish all the users at the same time.

We first characterize the set of sustainable action profiles, $\mathcal{E}$.
The following class of intervention rules is useful to characterize $\mathcal{E}$.
\begin{definition}
$f_{\tilde{\mathbf{a}}}:A \rightarrow A_0$ is an \emph{extreme intervention
rule with target action profile $\tilde{\mathbf{a}} \in A$} if
\begin{align} \label{eq:polar}
f_{\tilde{\mathbf{a}}}(\mathbf{a}) = \left\{
\begin{array}{ll}
\underline{a}_0 \quad &\textrm{if $\mathbf{a} = \tilde{\mathbf{a}}$,}\\
\overline{a}_0 \quad &\textrm{otherwise}.
\end{array} \right.
\end{align}
\end{definition}

Note that an extreme intervention rule uses only the two extreme
points of $A_0$. With an extreme intervention rule, the intervention device
chooses the most preferred action for the users when they
follow the target action profile while choosing the least preferred
action when they deviate. Hence, an extreme intervention rule
provides the strongest incentive for sustaining a given target
action profile, which leads us to the following lemma.

\begin{lemma}
If $\mathbf{a}^* \in \mathcal{E}$, then $\mathbf{a}^* \in \mathcal{E}(f_{\mathbf{a}^*})$.
\end{lemma}

\begin{proof}
Suppose that $\mathbf{a}^* \in \mathcal{E}$. Then there exists an
intervention rule $f$ such that $u_i(f(\mathbf{a}^*), \mathbf{a}^*) \geq u_i(f(a_i, \mathbf{a}_{-i}^*),
a_i, \mathbf{a}_{-i}^*)$ for all $a_i \in A_i$, for all $i \in \mathcal{N}$. Then we
obtain $u_i(\underline{a}_0,\mathbf{a}^*) \geq u_i(f(\mathbf{a}^*), \mathbf{a}^*) \geq u_i(f(a_i, \mathbf{a}_{-i}^*),
a_i, \mathbf{a}_{-i}^*) \geq u_i(\overline{a}_0, a_i, \mathbf{a}_{-i}^*)$ for all $a_i
\in A_i$, for all $i \in \mathcal{N}$, where the first and the third
inequalities follow from \eqref{eq:a0order}.
\end{proof}

Let $F^e$ be the set of all extreme intervention rules, i.e., $F^e =
\{ f_{\tilde{\mathbf{a}}} \in F : \tilde{\mathbf{a}} \in A \}$. Also, define
$\mathcal{E}^e = \cup_{f \in F^e} \mathcal{E}(f) = \{ \mathbf{a} \in A :
\exists f \in F^e \text{ such that $f$ sustains $\mathbf{a}$} \}$. By
applying Lemma 1, we can obtain the following results.

\begin{theorem} \label{prop:char}
(i) $\mathbf{a}^* \in \mathcal{E}$ if and only if $u_i(\underline{a}_0,\mathbf{a}^*)
\geq u_i(\overline{a}_0, a_i, \mathbf{a}_{-i}^*)$ for all $a_i \in A_i$, for
all $i \in \mathcal{N}$.\\
(ii) $\mathcal{E} = \mathcal{E}^e$.\\
(iii) If $(f^*,\mathbf{a}^*)$ is an intervention equilibrium, then $(f_{\mathbf{a}^*},
\mathbf{a}^*)$ is also an intervention equilibrium.
\end{theorem}

\begin{proof}
(i) Suppose that $u_i(\underline{a}_0,\mathbf{a}^*) \geq u_i(\overline{a}_0,
a_i, \mathbf{a}_{-i}^*)$ for all $a_i \in A_i$, for all $i \in \mathcal{N}$. Then
$f_{\mathbf{a}^*}$ sustains $\mathbf{a}^*$, and thus $\mathbf{a}^* \in \mathcal{E}$. The
converse follows from Lemma 1.

(ii) $\mathcal{E} \supset \mathcal{E}^e$ follows from $F \supset
F^e$, while $\mathcal{E} \subset \mathcal{E}^e$ follows from Lemma
1.

(iii) Suppose that $(f^*,\mathbf{a}^*)$ is an intervention equilibrium. Then
by Definition~\ref{def:ie}, $f^*$ sustains $\mathbf{a}^*$, and
$v_0^{f^*}(\mathbf{a}^*) \geq v_0^f(\mathbf{a})$ for all $(f,\mathbf{a}) \in {F} \times A$ such
that $\mathbf{a} \in \mathcal{E}(f)$. Since $\mathbf{a}^* \in \mathcal{E}$,
$\mathbf{a}^* \in \mathcal{E}(f_{\mathbf{a}^*})$ by Lemma 1. Hence, $v_0^{f^*}(\mathbf{a}^*) \geq
u_0(f_{\mathbf{a}^*}(\mathbf{a}^*),\mathbf{a}^*)$. On the other hand, since $f_{\mathbf{a}^*}(\mathbf{a}^*) =
\underline{a}_0$, we have $v_0^{f^*}(\mathbf{a}^*) \leq
u_0(f_{a^*}(\mathbf{a}^*),\mathbf{a}^*)$ by \eqref{eq:a0order}. Therefore,
$v_0^{f^*}(\mathbf{a}^*) = u_0(f_{\mathbf{a}^*}(\mathbf{a}^*),\mathbf{a}^*)$, and thus
$u_0(f_{\mathbf{a}^*}(\mathbf{a}^*),\mathbf{a}^*) \geq v_0^f(\mathbf{a})$ for all $(f,\mathbf{a}) \in {F} \times
A$ such that $\mathbf{a} \in \mathcal{E}(f)$. This proves that $(f_{\mathbf{a}^*}, \mathbf{a}^*)$ is
an intervention equilibrium.
\end{proof}

Theorem 1 shows that there is no loss of generality in three senses
when we restrict attention to extreme intervention rules. First, in
order to test whether there exists an intervention rule that
sustains a given action profile, it suffices to consider only the
extreme intervention rule having the action profile as its target
action profile. Second, the set of action profiles that can be
sustained by an intervention rule remains the same when we consider
only extreme intervention rules. Third, if there exists an optimal
intervention rule, we can find an optimal intervention rule among
extreme intervention rules.

Note that the role of extreme intervention rules is analogous to
that of trigger strategies in repeated games with perfect
monitoring. To generate the set of equilibrium payoffs, it suffices
to consider trigger strategies that trigger the most severe
punishment in case of a deviation. Under Assumption~\ref{ass:permon}, the
maximal intervention action $\overline{a}_0$ plays a similar role to
mutual minmaxing \cite{mailath} in that it provides
the strongest threat to deter a deviation. The next theorem provides
a necessary and sufficient condition under which an extreme
intervention rule together with its target action profile
constitutes an intervention equilibrium.

\begin{theorem} \label{prop:charie}
$(f_{\mathbf{a}^*}, \mathbf{a}^*)$ is an intervention equilibrium if and only if $\mathbf{a}^*
\in \mathcal{E}$ and $u_0(\underline{a}_0, \mathbf{a}^*) \geq
u_0(\underline{a}_0, \mathbf{a})$ for all $\mathbf{a} \in \mathcal{E}$.
\end{theorem}

\begin{proof}
Suppose that $(f_{\mathbf{a}^*}, \mathbf{a}^*)$ is an intervention equilibrium. Then
$f_{\mathbf{a}^*}$ sustains $\mathbf{a}^*$, and thus $\mathbf{a}^* \in \mathcal{E}$. Also,
$u_0(f_{\mathbf{a}^*}(\mathbf{a}^*),\mathbf{a}^*) \geq v_0^f(\mathbf{a})$ for all $(f,\mathbf{a}) \in F \times A$
such that $\mathbf{a} \in \mathcal{E}(f)$. Choose any $\mathbf{a} \in \mathcal{E}$. Then by
Lemma 1, $f_\mathbf{a}$ sustains $\mathbf{a}$, and thus $u_0(\underline{a}_0, \mathbf{a}^*) =
u_0(f_{\mathbf{a}^*}(\mathbf{a}^*),\mathbf{a}^*) \geq u_0(f_\mathbf{a}(\mathbf{a}),\mathbf{a}) = u_0(\underline{a}_0, \mathbf{a})$.

Suppose that $\mathbf{a}^* \in \mathcal{E}$ and $u_0(\underline{a}_0, \mathbf{a}^*)
\geq u_0(\underline{a}_0, \mathbf{a})$ for all $\mathbf{a} \in \mathcal{E}$. To prove
that $(f_{\mathbf{a}^*},\mathbf{a}^*)$ is an intervention equilibrium, we need to show
(i) $f_{\mathbf{a}^*}$ sustains $\mathbf{a}^*$, and (ii) $u_0(f_{\mathbf{a}^*}(\mathbf{a}^*),\mathbf{a}^*) \geq
v_0^f(\mathbf{a})$ for all $(f,\mathbf{a}) \in {F} \times A$ such that
$\mathbf{a} \in \mathcal{E}(f)$. Since $\mathbf{a}^* \in \mathcal{E}$, (i) follows from Lemma 1. To prove
(ii), choose any $(f,\mathbf{a}) \in {F} \times A$ such that $\mathbf{a} \in \mathcal{E}(f)$.
Then $u_0(f_{\mathbf{a}^*}(\mathbf{a}^*),\mathbf{a}^*) = u_0(\underline{a}_0, \mathbf{a}^*) \geq
u_0(\underline{a}_0, \mathbf{a}) \geq v_0^f(\mathbf{a})$, where the first inequality
follows from $\mathbf{a} \in \mathcal{E}$.
\end{proof}

Theorem~\ref{prop:charie} implies that if we obtain an action
profile $\mathbf{a}^*$ such that $\mathbf{a}^* \in \arg \max_{\mathbf{a} \in \mathcal{E}}
u_0(\underline{a}_0, \mathbf{a})$, we can use it to construct an intervention
equilibrium and thus an optimal intervention rule.

\section{Illustrative Example}

In this section, we discuss an example to illustrate the
results in Section 3. Consider a wireless network with
two users and an intervention device interfering with each other. The action of
user $i$ is its usage level, where $A_i = [0, \overline{a}_i]$ for
$i = 0,1,2$. $\overline{a}_i$ is the maximum usage level of user
$i$. The total usage level is given by $a_0 + a_1 + a_2$. The
quality of service is determined by the total usage level, following
the relationship
\begin{align}
Q(a_0, a_1, a_2) = [q - b(a_0 + a_1 + a_2)]^+,
\end{align}
where $q,b > 0$ and $[x]^+ = \max \{x, 0\}$. The payoff of
user $i \in \{1,2\}$ is given by the product of the quality received and its
usage level,
\begin{align} \label{eq:payoffexample}
u_i(a_0, a_1, a_2) = Q(a_0, a_1, a_2) a_i.
\end{align}
The system objective is given by social welfare, which is defined as
the sum of the payoffs of the users,
\begin{align}
u_0(a_0, a_1, a_2) = u_1(a_0, a_1, a_2) + u_2(a_0, a_1, a_2).
\end{align}
Note that if there is no intervention device (i.e., if $a_0$ is held
fixed at 0), the example is identical to the Cournot duopoly model with
a linear demand function and zero production cost. The corresponding
Cournot duopoly game achieves the symmetric social optimum at $a_1 =
a_2 = a_L := q/4b$ while it has the unique Cournot-Nash equilibrium
at $a_1 = a_2 = a_H := q/3b$, as depicted in
Figure~\ref{fig:cournot}. Hence, the goal of the manager is to improve upon
the inefficient outcome $(a_H,a_H)$ by introducing the intervention device in the network.

\begin{figure}%
\centering \hspace{-15mm}
\includegraphics[width=0.5\textwidth]{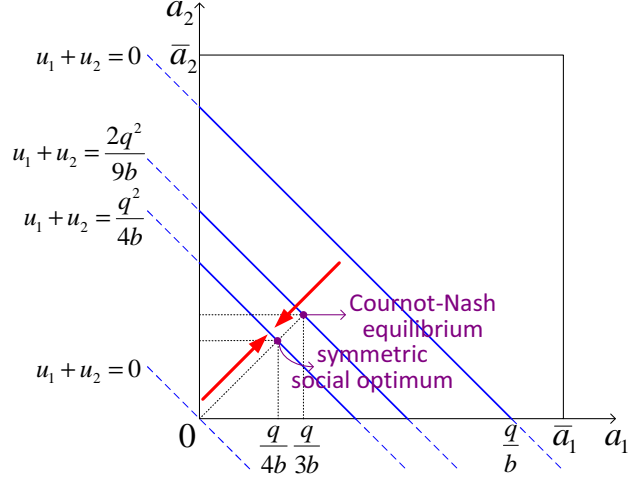}
\caption{Contour lines of social welfare in the Cournot duopoly
game.} \label{fig:cournot}
\end{figure}

Given the structure of the intervention game in this example,
the capability of the intervention device is determined by its
maximum intervention level $\overline{a}_0$. In the following, we
investigate sustainable action profiles and those that constitute an intervention
equilibrium as we vary $\overline{a}_0$.

\begin{proposition}
(i) If $\overline{a}_0 = 0$, then $\mathcal{E} = \{(a_H,a_H)\}$.\\
(ii) If $\overline{a}_0 \geq q/b$, then $\mathcal{E} = A$.\\
(iii) If $\overline{a}_0 \geq (3\sqrt{2} - 4)q/ 4\sqrt{2} b$, then
$\{(a_L,a_L)\} \in \mathcal{E}$ and thus $\{(a_L,a_L)\}$ constitutes
an intervention equilibrium.
\end{proposition}

If the intervention device cannot affect the payoffs of the users
($\overline{a}_0 = 0$), the non-cooperative outcome $(a_H,a_H)$
is the only sustainable action profile that is consistent with the
self-interest of the users. On the other hand, if the intervention device can apply a
sufficiently high intervention level ($\overline{a}_0 \geq q/b$),
it has the ability to degrade the quality to zero no matter what
action profile the users choose. Since the payoffs of the
users are non-negative, the punishment from using $\overline{a}_0$
is strong enough to make every action profile sustainable. We can also find a
condition on $\overline{a}_0$ that enables $f_{(a_L,a_L)}$ to
sustain the symmetric social optimum $(a_L,a_L)$. With
$\overline{a}_0 \geq (3\sqrt{2} - 4)q/ 4\sqrt{2} b$, $(a_L,a_L)$ is
sustainable and thus $(f_{(a_L,a_L)}, (a_L,a_L))$ is an
intervention equilibrium by Theorem 2.

Figure~\ref{fig:a0vary} plots the set $\mathcal{E}$ for six different
values of $\overline{a}_0$ with parameters $q=12, b=1$, and
$\overline{a}_1 = \overline{a}_2 = 12$. We can see that
$\mathcal{E}$ expands as $\overline{a}_0$ increases, starting from a
single point $(a_H,a_H) = (4,4)$ when $\overline{a}_0 = 0$ to the entire space $A$ when
$\overline{a}_0 \geq q/b= 12$. When $\overline{a}_0 < (3\sqrt{2} -
4)q/ 4\sqrt{2} b \approx 0.51$, only the action profile that is
closest to $(a_L,a_L)=(3,3)$ among those in $\mathcal{E}$ constitutes an
intervention equilibrium. When $\overline{a}_0 \geq (3\sqrt{2} -
4)q/ 4\sqrt{2} b \approx 0.51$, the action profiles in $\mathcal{E}$
that satisfies $a_1 + a_2 = 2a_L = 6$ constitute an intervention
equilibrium, as all of them yield the maximum social welfare.

\begin{figure}%
\centering
\subfloat[][$\overline{a}_0 = 0$]{%
\label{fig:a0vary-a}%
\hspace{-10mm}
\includegraphics[width=0.4\textwidth]{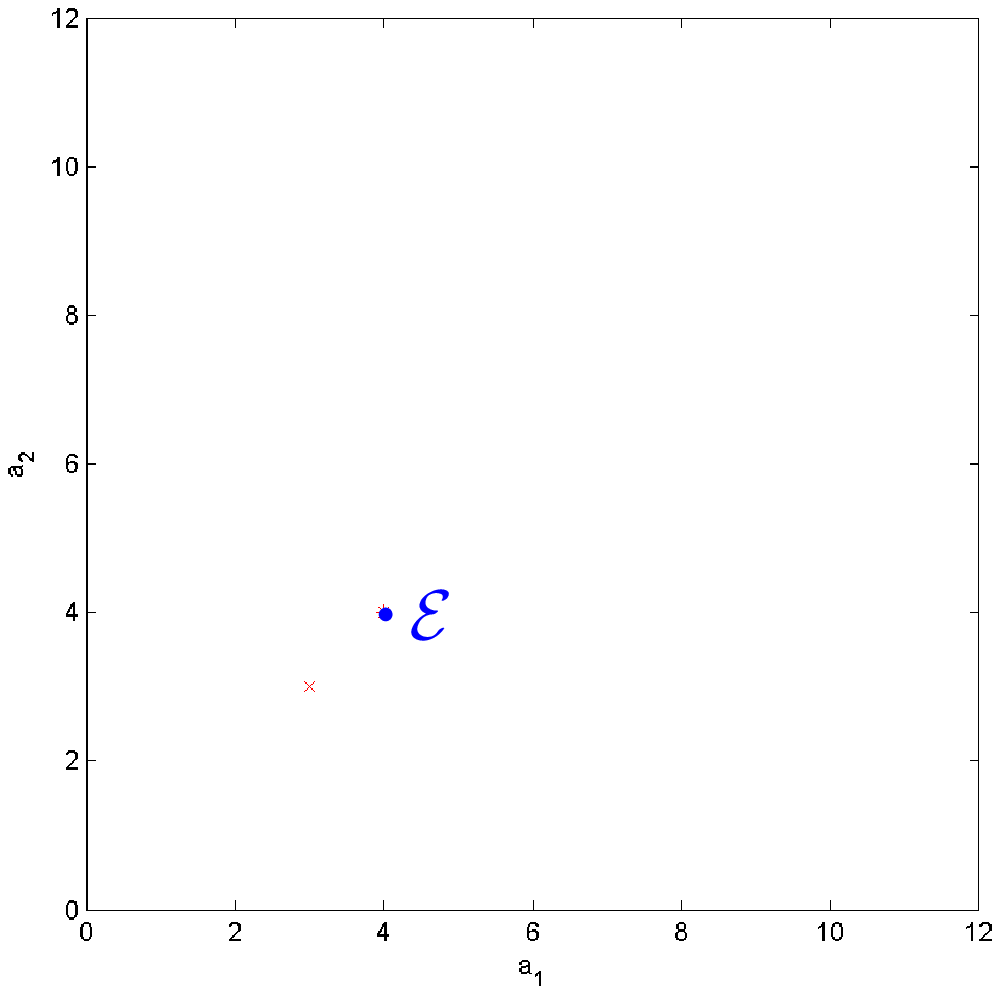}}%
\hspace{14mm}
\subfloat[][$\overline{a}_0 = 0.1$]{%
\label{fig:a0vary-b}%
\includegraphics[width=0.4\textwidth]{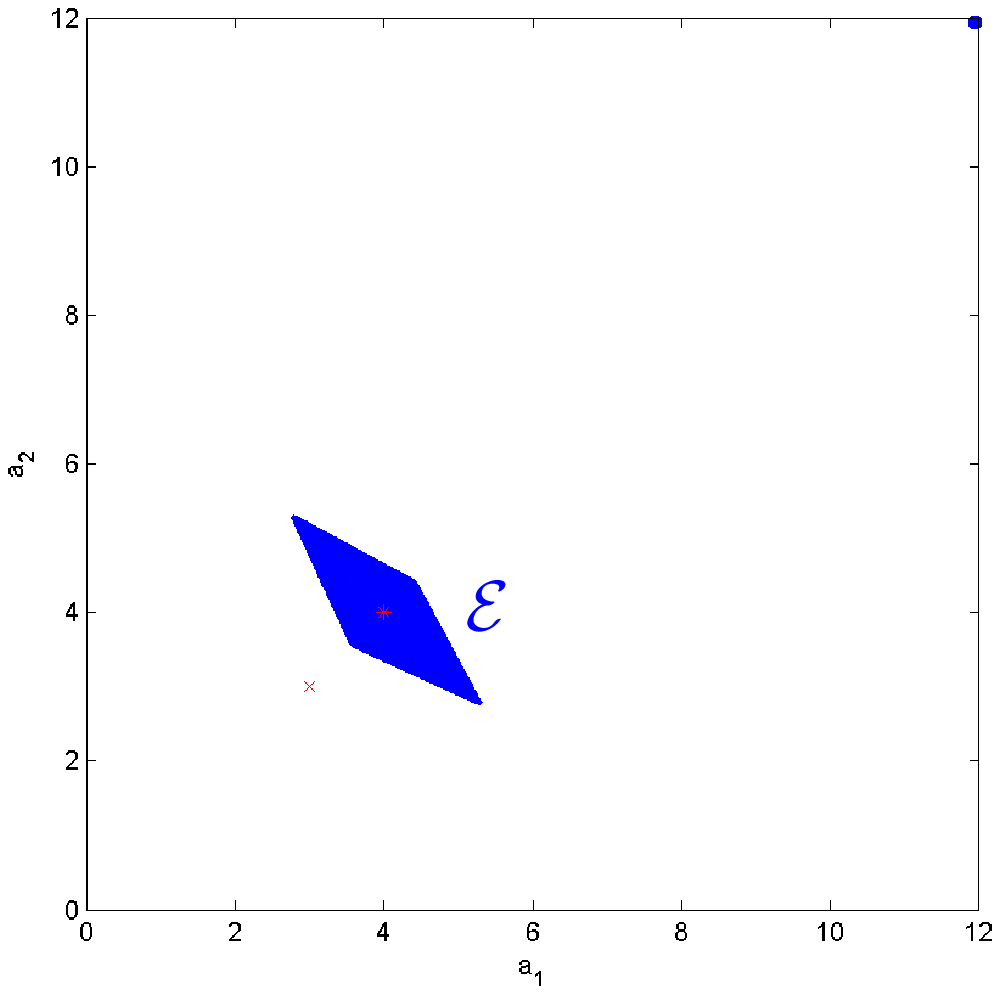}}\\
\subfloat[][$\overline{a}_0 = 0.51$]{%
\label{fig:a0vary-c}%
\hspace{-10mm}
\includegraphics[width=0.4\textwidth]{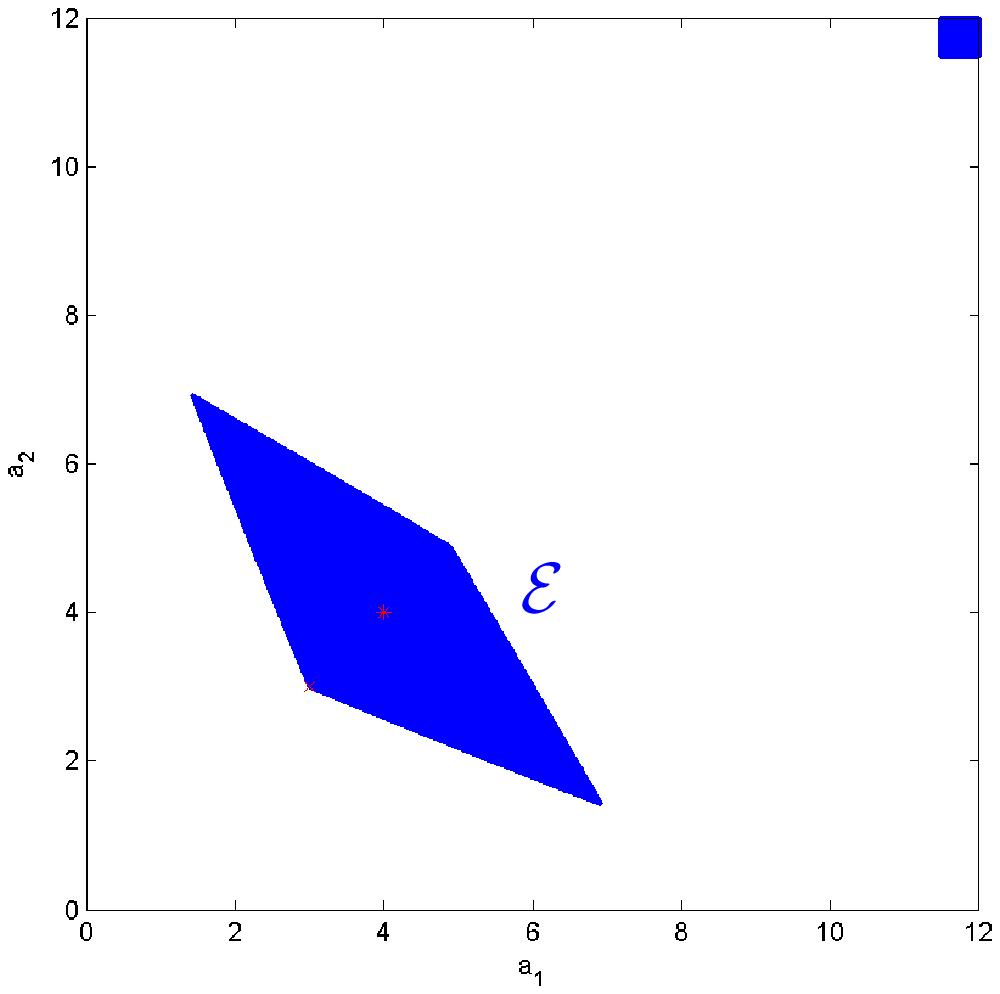}}%
\hspace{14mm}
\subfloat[][$\overline{a}_0 = 5$]{%
\label{fig:a0vary-d}%
\includegraphics[width=0.4\textwidth]{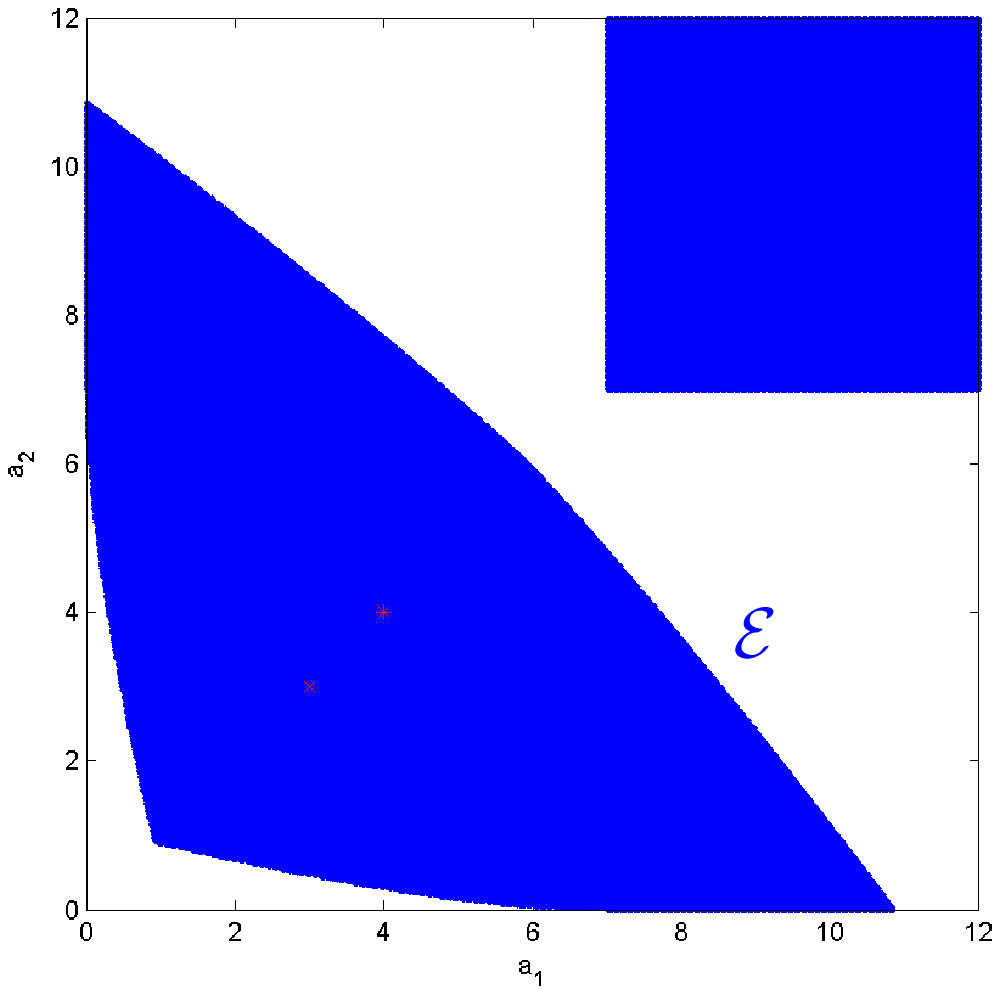}}\\
\subfloat[][$\overline{a}_0 = 10$]{%
\label{fig:a0vary-e}%
\hspace{-10mm}
\includegraphics[width=0.4\textwidth]{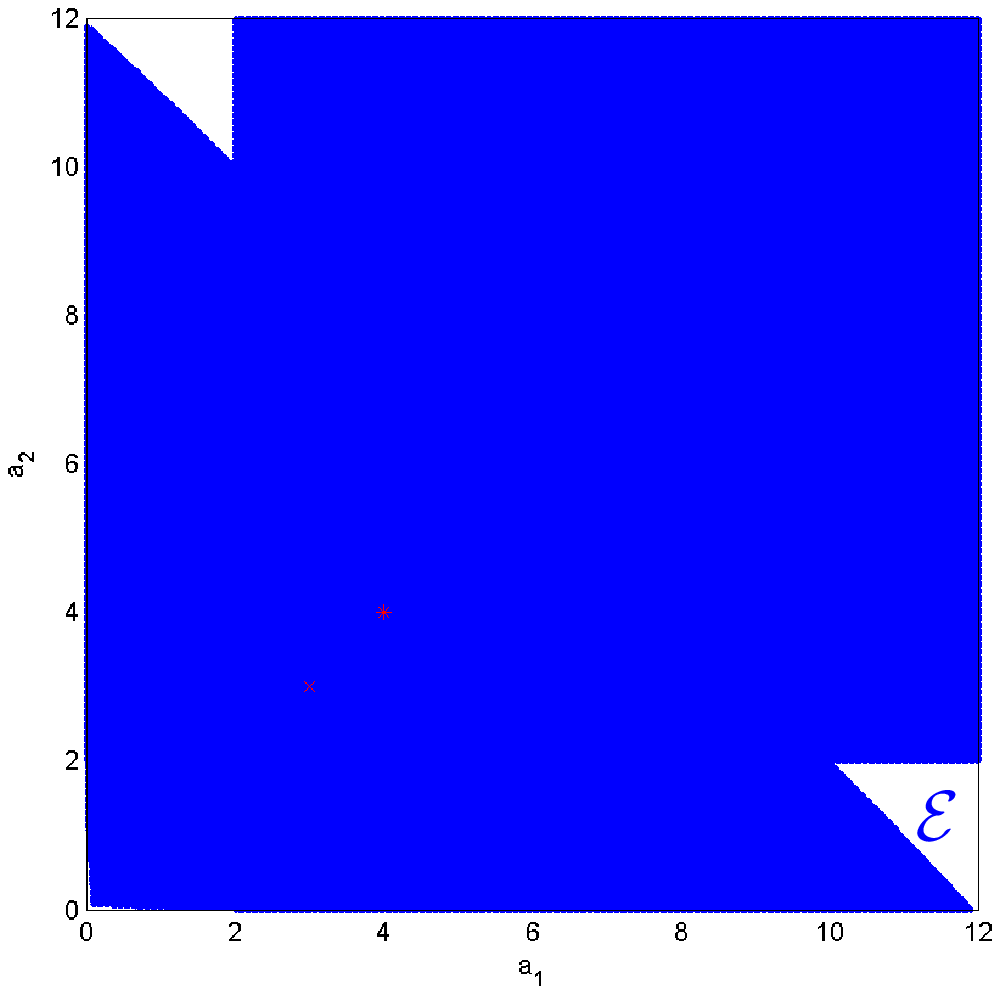}}%
\hspace{14mm}
\subfloat[][$\overline{a}_0 \geq 12$]{%
\label{fig:a0vary-f}%
\includegraphics[width=0.4\textwidth]{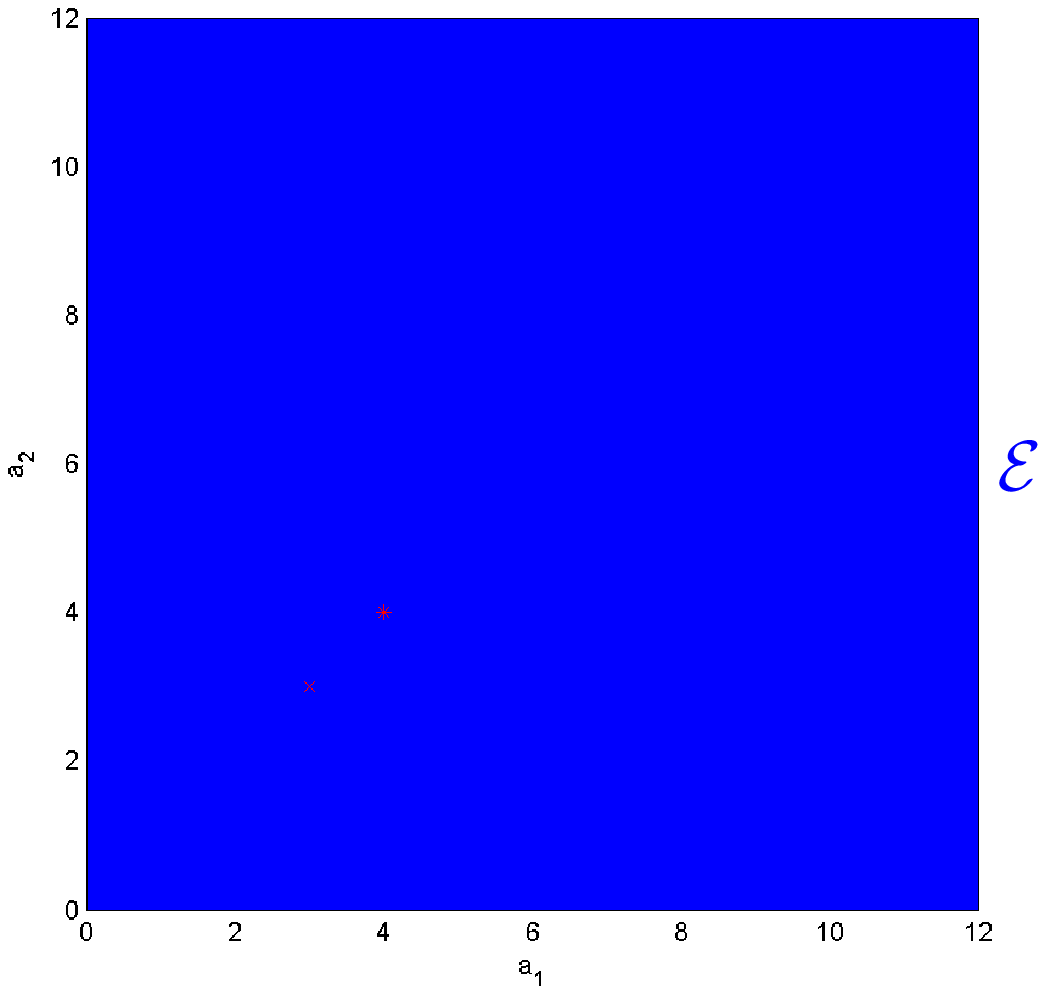}}\\
\caption{Shape of $\mathcal{E}$ for the different values of
$\overline{a}_0$.}
\label{fig:a0vary}%
\end{figure}


\small
\singlespacing

\end{document}